\documentclass[11pt,a4paper 
]{article}
\usepackage{amsmath,amssymb,latexsym}
\usepackage [utf8]{inputenc}
\usepackage[T2A]{fontenc}
\usepackage[left=2cm,right=2cm, top=1cm,bottom=1cm]{geometry}
\usepackage[colorlinks=true]{hyperref}
\usepackage{amsthm}
\usepackage{enumitem}
\usepackage{bm,bbm}

\numberwithin{equation}{section}
\newcommand{\Compl}{\mathbb{C}}
\newcommand{\R}{\mathbb{R}}
\newcommand{\Z}{\mathbb{Z}}

\newcommand{\conj}[1]{\overline{#1}}
\newcommand{\cConjScl}[1]{\bar{#1}}
\newcommand{\aConjScl}[1]{#1^*}
\newcommand{\aConjMat}[1]{#1^+}

\newcommand{\charfp}[2]{\psi\left(#1,#2\right)}
\newcommand{\asKer}[2]{K\left(#1, #2\right)}
\newcommand{\CF}{\mathsf{f}}

\newcommand{\ens}{M_n}
\newcommand{\sclA}{y}
\newcommand{\matA}{Y}

\newcommand{\cMatGin}{X}
\newcommand{\cMatPos}{\mathcal{Z}}
\newcommand{\cSclN}{t}
\newcommand{\cVecN}{\bm{\cSclN}}
\newcommand{\cSclA}{q}
\newcommand{\cMatA}{Q}
\newcommand{\cSetMatA}{\bm{Q}}
\newcommand{\cVecGI}{\bm{h}}

\newcommand{\aSclA}{\xi}

\newcommand{\aMatA}{\Xi}
\newcommand{\aSetMatA}{\mathbf{\aMatA}}
\newcommand{\aSclB}{\phi}
\newcommand{\aVecB}{\bm{\aSclB}}
\newcommand{\aMatB}{\Phi}
\newcommand{\aSclBt}{\varphi}
\newcommand{\aVecBt}{\bm{\aSclBt}}
\newcommand{\aSclC}{\theta}
\newcommand{\aVecC}{\bm{\aSclC}}
\newcommand{\aMatC}{\Theta}
\newcommand{\aSclCt}{\vartheta}
\newcommand{\aVecCt}{\bm{\aSclCt}}
\newcommand{\aSclD}{\rho}
\newcommand{\aVecD}{\bm{\aSclD}}
\newcommand{\aSclE}{\tau}
\newcommand{\aSclF}{\nu}
\newcommand{\aVecF}{\bm{\aSclF}}
\newcommand{\aSclG}{\upsilon}

\newcommand{\aMatG}{\Upsilon}
\newcommand{\aSclH}{\aSclE}
\newcommand{\aVecH}{\bm{\aSclH}}
\newcommand{\aVecGrI}{\bm{\upsilon}}

\newcommand{\cumul}[2]{\kappa_{#1,#2}}
\newcommand{\indexset}{\mathcal{I}}

\newcommand{\stpointsnbh}{\Omega_n}
\newcommand{\idom}{\mathcal{D}}
\newcommand{\Vanddet}{\triangle}
\newcommand{\herm}{\mathcal{H}}

\newcommand{\der}[2]{\frac{d #1}{d #2}}

\newcommand{\abs}[1]{\left\lvert#1\right\rvert}

\newcommand{\norm}[1]{\left\lVert#1\right\rVert}
\newcommand{\normsized}[2][ ]{#1\lVert#2#1\rVert}

\newcommand{\acknowledgement}[1]{\medskip
	\textbf{Acknowledgement.} #1}

\DeclareMathOperator{\tr}{tr}
\DeclareMathOperator{\E}{\mathbf{E}}

\DeclareMathOperator{\diag}{diag}

\newtheorem{thm}{Theorem}
\newtheorem{prop}{Proposition}
\newtheorem{lem}{Lemma}

\theoremstyle{remark}
\newtheorem{rem}{Remark}

\title{On the Correlation Functions of the Characteristic Polynomials of Non-Hermitian Random Matrices with Independent Entries}
\author{Ie.\ Afanasiev\footnote{B. Verkin Institute for Low Temperature Physics and Engineering of the National Academy of Sciences of Ukraine, Kharkiv, Ukraine;
e-mail: \href{mailto:afanasiev@ilt.kharkov.ua}{afanasiev@ilt.kharkov.ua}. Supported in part by The President of Ukraine grant and by the Akhiezer Foundation scholarship.}
}

\begin{document}

\maketitle
\begin{abstract}
	The paper is concerned with the asymptotic behavior of the correlation functions of the characteristic polynomials of non-Hermitian random matrices with independent entries. It is shown that the correlation functions behave like that for the Complex Ginibre Ensemble up to a factor depending only on the fourth absolute moment of the common probability law of the matrix entries.
\end{abstract}
\section{Introduction}
The Random Matrix Theory has been developed for some sixty years. The story began with the study of symmetric and Hermitian random matrices. They have remained the most studied ever since. However, non-Hermitian matrices are not so well studied.

The present paper is concerned with the simplest non-Hermitian ensemble which is an analog of the Wigner ensemble. The matrices are constructed of independent identically distributed (i.i.d.)\ complex random variables. More precisely, the matrices have the form
\begin{equation} \label{Gin ens}
\ens = \frac{1}{\sqrt{n}}\cMatGin = \frac{1}{\sqrt{n}}(x_{jk})_{j,k = 1}^n,
\end{equation}
where $x_{jk}$ are i.i.d.\ complex random variables such that
\begin{equation}\label{moments}
\E\{x_{jk}\} = \E\{x_{jk}^2\} = 0, \quad \E\{\abs{x_{jk}}^2\} = 1.
\end{equation}
Here and everywhere below $\E$ denotes the expectation with respect to (w.r.t.)\ all random variables. This ensemble has various applications in physics, neuroscience, economics, etc. For detailed information see \cite{Ak-Ka:07} and references therein.

Define the Normalized Counting Measure (NCM) of eigenvalues as
\begin{equation*}
	N_n(\Delta) = \#\{\lambda_j^{(n)} \in \Delta,\, j = 1, \ldots, n\}/n, 
\end{equation*}
where $\Delta$ is an arbitrary Borel set in the complex plane, $\left\{\lambda_j^{(n)}\right\}_{j = 1}^n$ are the eigenvalues of~$\ens$. The NCM is known to converge to the uniform distribution on the unit disc. This distribution is called the circular law. This result has a long and rich history. Mehta was the first who obtained it for $x_{jk}$ being complex Gaussian in 1967 \cite{Me:67}. The proof strongly relied on the explicit formula for the common probability density of eigenvalues due to Ginibre \cite{Gin:65}. Unfortunately, there is no such a formula in the general case. That is why other methods have to be used. The Hermitization approach introduced by Girko \cite{Gir:84} appeared to be an effective method. The main idea is to reduce the study of matrices~\eqref{Gin ens} to the study of Hermitian matrices using the logarithmic potential of a measure
\begin{equation*}
	P_\mu(z) = \int\limits_\Compl \log \abs{z - \zeta}\, d\mu(\zeta).
\end{equation*}
This approach was successfully developed by Girko in the next series of works \cite{Gir:94, Gir:04_1, Gir:04_2, Gir:05}. The final result in the most general case was established by Tao and Vu \cite{Ta-Vu:10}. Notice that there are a lot of partial results besides those listed above. The interested reader is directed to \cite{Bo-Ch:12}.

The Central Limit Theorem (CLT) for non-Hermitian random matrices linear statistics was proven in some partial cases in \cite{Fo:99, Ri-Si:06, Ri-Vi:07}. The best results for today were obtained by Kopel in \cite{Ko:15} for smooth functions and by Tao and Vu in \cite{Ta-Vu:15} for small radii disc indicators. Both mentioned results require $\Re x_{jk}$ and $\Im x_{jk}$ being independent and having the first four moments as in the Gaussian case (which is often referred as GinUE similarly to the Gaussian Unitary Ensemble (GUE) in Hermitian case). The article \cite{Ta-Vu:15} also deals with a local regime for these matrices. It was established that under the same conditions the $k$-point correlation function converges in vague topology to that for GinUE.

One can observe that non-Hermitian random matrices are more complicated than their Hermitian counterparts. Indeed, the Hermitian case was successfully dealt with using the Stieltjes transform or the moments method. However, a measure in the plane can not be recovered from its Stieltjes transform or its moments. Thus these approaches to the analysis fail in the non-Hermitian case. 

The present article suggests to apply the supersymmetry technique (SUSY). It is a rather powerful method which is widely applied at the physical level of rigor (for instance \cite{Fy-Mi:91, Mi-Fy:91}). 
There are also a lot of rigorous results, which were obtained using SUSY in the recent years, e.g.\ \cite{Di-Lo-So:18}, \cite{Di-Sp-Zi:10}, \cite{ShM-ShT:16}, \cite{ShM-ShT:17}, \cite{ShM-ShT:18}, etc. Supersymmetry technique is usually used in order to obtain an integral representation for ratios of determinants.
Since the main spectral characteristics such as density of states, spectral correlation functions, etc.\ often can be expressed via ratios of determinants, SUSY allows to get the integral representation for these characteristics too. For detailed discussion on connection between spectral characteristics and ratios of determinants see \cite{St-Fy:03, Bo-St:06, Gu:15}. See also \cite{Fy-St:03, Re-Ki-Gu-Zi:12}.

Let us consider the second spectral correlation function $R_2$ defined by the equality
\begin{equation*}
	\E\Bigg\{2\sum\limits_{1 \le j_1 < j_2 \le n}\eta\left(\lambda_{j_1}^{(n)}, \lambda_{j_2}^{(n)}\right)\Bigg\} = \int\limits_{\Compl^2} \eta(\lambda_1, \lambda_2) R_2(\lambda_1, \lambda_2) d\bar{\lambda}_1 d\lambda_1 d\bar{\lambda}_2 d\lambda_2,
\end{equation*}
where the function $\eta \colon \Compl^2 \to \Compl$ is bounded, continuous and symmetric in its arguments. Using the logarithmic potential, $R_2$ can be represented via ratios of the determinants of $\ens$ with the most singular term of the form
\begin{equation}\label{2detsratio}
\int\limits_{0}^{\varepsilon_0}\int\limits_{0}^{\varepsilon_0}\frac{\partial^2}{\partial \delta_1 \partial \delta_2}\E\Bigg\{
\prod\limits_{j = 1}^2 \frac{\det \left((\ens - z_j)(\ens - z_j)^* + \delta_j\right)}{\det \left((\ens - z_j)(\ens - z_j)^* + \varepsilon_j\right)}
\Bigg\}\Bigg|_{\delta = \varepsilon} d\varepsilon_1 d\varepsilon_2
\end{equation}
The integral representation for \eqref{2detsratio} obtained by SUSY will contain both commuting and anti-commuting variables. Such type integrals are rather difficult to analyse. That is why one would investigate a more simple but similar integral to shed light on the situation. This integral arises from the study of the correlation functions of the characteristic polynomials. Moreover, the correlation functions of the characteristic polynomials are of independent interest. They were studied for many ensembles of Hermitian and real symmetric matrices, for instance \cite{Br-Hi:00}, \cite{Br-Hi:01}, \cite{ShT:11}, \cite{ShT:13}, \cite{ShM-ShT:17} etc. The other result on the asymptotic behavior of the correlation functions of the characteristic polynomials of non-Hermitian matrices of the form $H + i\Gamma$, where $H$ is from GUE and $\Gamma$ is a fixed matrix of rank $M$, was obtained in \cite{Fy-Kh:99}. The kernel computed there, in the limit of rank $M \to \infty$ of the perturbation $\Gamma$ (taken after matrix size $n \to \infty$) after appropriate rescaling approaches the form \eqref{asympt ker}. It was demonstrated in \cite[Sec. 2.2]{Fy-So:03}.

Let us introduce the $m$\textsuperscript{th} correlation function of the characteristic polynomials
\begin{equation}
\label{F_m}
\CF_m(Z) = \E\Bigg\{
\prod\limits_{j = 1}^m \det \left(\ens - z_j\right)\left(\ens - z_j\right)^*
\Bigg\},
\end{equation}
where 
\begin{equation}\label{Z def}
Z = \diag\{z_1, \dotsc, z_m\}
\end{equation}
and $z_1$, \ldots, $z_m$ are complex parameters which may depend on~$n$. We are interested in the asymptotic behavior of \eqref{F_m}, as $n \to \infty$, for
\begin{equation}\label{z_j}
z_j = z_0 + \frac{\zeta_j}{\sqrt{n}}, \quad j = 1,2,\dotsc,m,
\end{equation}
where $z_0$, $\zeta_1$, \ldots, $\zeta_m$ are $n$-independent complex numbers, and $z_0$ is in the bulk of the spectrum, i.e. $\abs{z_0} < 1$. For GinUE the value of \eqref{F_m} is known. In \cite{Ak-Ve:03} Akemann and Vernizzi showed that
\begin{equation}\label{finite-n Gauss result}
	\E\Bigg\lbrace \prod\limits_{j = 1}^m \det \left(X - z_j^{(1)}\right)\left(X - z_j^{(2)}\right)^*
	\Bigg\rbrace = \Bigg(\prod\limits_{l = n}^{n + m - 1} l!\Bigg) \frac{\det (K_n(z_j^{(1)}, \cConjScl{z}_k^{(2)}))_{j,k = 1}^m}{\Vanddet(Z^{(1)})\Vanddet((Z^{(2)})^*)},
\end{equation}
where
\begin{equation*}
	K_n(z, w) = \sum\limits_{l = 0}^{n + m - 1} \frac{(z\cConjScl{w})^l}{l!}.
\end{equation*}
and $\Vanddet(Z^{(1)})$ (resp. $\Vanddet((Z^{(2)})^*)$) is a Vandermonde determinant of $z_1^{(1)}$, \ldots, $z_m^{(1)}$ (resp. $\cConjScl{z}_1^{(2)}$, \ldots, $\cConjScl{z}_m^{(2)}$). Putting $z_j^{(1)} = z_j^{(2)} = \sqrt{n}z_j$ , $z_j$ of the form \eqref{z_j}, one deduces 
from \eqref{finite-n Gauss result} that
\begin{equation}\label{asympt: Gauss case}
	\lim\limits_{n \to \infty} n^{-\frac{m^2 - m}{2}}\frac{\CF_m(Z)}{\CF_1(z_1)\dotsb \CF_1(z_m)} = \frac{\det (K(\zeta_j, \zeta_k))_{j,k = 1}^m}{\abs{\Vanddet(\cMatPos)}^2},
\end{equation}
where $\cMatPos = \diag\{\zeta_1, \dotsc, \zeta_m\}$ and
\begin{equation}\label{asympt ker}
	\asKer{z}{w} = e^{-\abs{z}^2/2 - \abs{w}^2/2 + z\bar{w}}.
\end{equation}

The other result on the characteristic polynomials of GinUE matrices was obtained by Webb and Wong in \cite{We-Wo:19}. They showed that for any complex $\gamma$ with $\Re \gamma > -2$
\begin{equation}\label{gamma-moment}
	\E\left\lbrace \abs{\det (M_n - z_0)}^{\gamma} \right\rbrace = n^{\frac{\gamma^2}{8}} e^{\frac{\gamma}{2}n(\abs{z_0}^2 - 1)}\frac{(2\pi)^{\frac{\gamma}{4}}}{G(1 + \frac{\gamma}{2})}(1 + o(1)),
\end{equation}
where $G$ is the Barnes $G$-function.

In this article the general case of arbitrary distribution, satisfying \eqref{moments} is considered. The main result of the paper is
\begin{thm}\label{th1}
	Let an ensemble of non-Hermitian random matrices $\ens$ be defined by~\eqref{Gin ens} and~\eqref{moments}. Let also the first $2m$ absolute moments of the common distribution of entries of $\ens$ be finite and $z_j$, $j = 1, \dotsc, m$, have the form \eqref{z_j}. Then
	\begin{enumerate}[label=(\roman*)]
	\item the $m$\textsuperscript{th} correlation function of the characteristic polynomials~\eqref{F_m} 
	satisfies the asymptotic relation
	\begin{equation*} 
		\lim\limits_{n \to \infty} n^{-\frac{m^2 - m}{2}}\frac{\CF_m(Z)}{\CF_1(z_1)\dotsb \CF_1(z_m)} = C_{m,z_0}^{(i)}
		e^{\frac{m^2 - m}{2}\left(1 - \abs{z_0}^2\right)^2\cumul{2}{2}} \frac{\det (K(\zeta_j, \zeta_k))_{j,k = 1}^m}{\abs{\Vanddet(\cMatPos)}^2},
	\end{equation*}
	where $C_{m,z_0}^{(i)}$ is some constant, which does not depend on the common distribution of entries and on $\zeta_1$, \ldots, $\zeta_m$; 
	$\cumul{2}{2} = \E\{\abs{x_{11}}^4\} - 2$ 
	and $\asKer{z}{w}$ is defined in \eqref{asympt ker};
	\item in particular case $\zeta_1 = \dotsb = \zeta_m = 0$ we have 
	\begin{equation}\label{main:moments}
		\E\left\lbrace \abs{\det (M_n - z_0)}^{2m} \right\rbrace = C_{m,z_0}^{(ii)} e^{\frac{m^2 - m}{2}\left(1 - \abs{z_0}^2\right)^2\cumul{2}{2}} n^{\frac{m^2}{2}} e^{mn(\abs{z_0}^2 - 1)}(1 + o(1)),
	\end{equation}
	where $C_{m,z_0}^{(ii)}$ is some constant, which does not depend on the common distribution of entries.
	\end{enumerate}
\end{thm}

\begin{rem}
	Going through the proof of Theorem \ref{th1} one can determine constants $C_{m,z_0}^{(i)}$ and $C_{m,z_0}^{(ii)}$. Their values are
	\begin{align*}
		&C_{m,z_0}^{(i)} = 1,
		&C_{m,z_0}^{(ii)} = (2\pi)^{m/2} \Bigg(\prod\limits_{j = 1}^{m - 1} j!\Bigg)^{-1}.
	\end{align*}
\end{rem}

Let us point out that $\cumul{2}{2} = 0$ in Gaussian case, and the result of Theorem~\ref{th1} is in full agreement with the results for GinUE (cf.\ \eqref{asympt: Gauss case}, \eqref{gamma-moment}). Theorem also shows that the asymptotics of $\CF_m$\footnote{Here and below we omit $Z$ only if $Z = \diag\{z_1, \dotsc, z_m\}$.} is similar to the asymptotics of the $m$-point spectral correlation function (see~\cite{Ta-Vu:15}).

The paper is organized as follows. Section~\ref{sec:IR} is devoted to the derivation of the suitable integral representation for $\CF_m$ by using the SUSY approach. In Section~\ref{sec:asympt analysis} we apply the steepest descent method to the obtained integral representation and find out the asymptotic behavior of $\CF_m$. In order to compute it, the Harish-Chandra/Itsykson--Zuber formula is used. For the reader convenience both sections are divided into two parts treating the Gaussian and general cases respectively.

\acknowledgement{The author is grateful to Prof.\ M.\ Shcherbina for the statement of the problem and fruitful discussions.
}

\subsection{Notations}
Through out the article lower-case letters denote scalars, bold lower-case letters denote vectors, upper-case letters denote matrices and bold upper-case letters denote sets of matrices. We use the same letter for a matrix, for its columns and for its entries. 
Table \ref{tab:notation} shows an exact correspondence.
\begin{table}
	\begin{center}
		\begin{tabular}{|c|c|c|c|}
			\hline
			Set of matrices & Matrix & Column & Entry \\
			\hline
			$\cSetMatA$ & $\cMatA_{p,s}$ & & $\cSclA^{(p,s)}_{\alpha\beta}$ \\
			\hline
			$\aSetMatA$ & $\aMatA_{p,s}$ & & $\aSclA^{(p,s)}_{\alpha\beta}$ \\
			\hline
			& $\aMatB$ & $\aVecB_j$ & $\aSclB_{kj}$ \\
			\hline
			& $\aMatC$ & $\aVecC_j$ & $\aSclC_{kj}$ \\
			\hline
			& $\matA_{k,p,s}$ & & $\sclA^{(k,p,s)}_{\alpha\beta}$ \\
			\hline
			& $U$ & & $u_{kj}$ \\
			\hline
			& $V$ & & $v_{kj}$ \\
			\hline
		\end{tabular}
	\end{center}
	\caption{Notation correspondence}\label{tab:notation}
\end{table}
Besides, for any matrix $A$ we denote by $(A)_j$ its $j$-th column and by $(A)_{kj}$ its entry in the $k$-th row and in the $j$-th column.

The term ``Grassmann variable'' is a synonym for ``anti-commuting variable''. The variables of integration $\aSclB$, $\aSclBt$, $\aSclC$, $\aSclCt$, $\aSclD$, $\aSclA$, $\aSclE$ and $\aSclF$ are Grassmann variables, all the other variables of integration unspecified by an integration domain are either complex or real. We split all the generators of Grassmann algebra into two equal sets and consider the generators from the second set as ``conjugates'' of that from the first set. I.e., for Grassmann variable $\aSclG$ we use $\aConjScl{\aSclG}$ to denote its ``conjugate''. Furthermore, if $\aMatG = (\aSclG_{jk})$ means a matrix of Grassmann variables then $\aConjMat{\aMatG}$ is a matrix $(\aConjScl{\aSclG_{kj}})$. $d$-dimensional vectors are identified with $d \times 1$~matrices.

Integrals without limits denote either integration over Grassmann variables or integration over the whole space $\Compl^d$ or $\R^d$. Let also $d\cVecN^*d\cVecN$ ($\cVecN = (\cSclN_1, \dotsc, \cSclN_d)^T \in \Compl^d$) denote the measure $\prod\limits_{j = 1}^d d\cConjScl{\cSclN}_j d\cSclN_j$ on the space $\Compl^d$. Similarly, for vectors with anti-commuting entries $d\aConjMat{\aVecH}d\aVecH = \prod\limits_{j = 1}^d d\aConjScl{\aSclH_j} d\aSclH_j$. Note that the space of matrices is a linear space over $\Compl$. Thus the same notations are used for them as well.

Through out the article $U(m)$ is a group of unitary $m \times m$ matrices. In order to simplify the notation we sometimes write $\cMatA_j$ instead of $\cMatA_{j,j}$ and $\cSclA_{\alpha\beta}^{(j)}$ instead of $\cSclA_{\alpha\beta}^{(j,j)}$. In addition, $C$, $C_1$~denote various $n$-independent constants which can be different in different formulas.

\section{Integral representation for $\CF_m$}\label{sec:IR}
In this section we obtain a convenient integral representation for the correlation function of characteristic polynomials $\CF_m$ defined by \eqref{F_m}.
\begin{prop}\label{prop:IR}
	Let an ensemble $\ens$ be defined by~\eqref{Gin ens} and~\eqref{moments}. Then 
	the $m$\textsuperscript{th} correlation function of the characteristic polynomials $\CF_m$ defined by \eqref{F_m} can be represented in the following form
	\begin{equation}\label{IR result}
	\CF_m = \left(\frac{n}{\pi}\right)^{c_m} \int g(\cSetMatA) e^{(n - c_m)f(\cSetMatA)} d\cSetMatA,
	\end{equation}
	where $c_m = 2^{2m - 1}$, $\cSetMatA = (\cMatA_{p,s})_{p,s = 1}^m$ with even $p + s$, $\cMatA_{p,s}$ is a complex $\binom{m}{p} \times \binom{m}{s}$ matrix, $d\cSetMatA = \prod\limits_{\substack{p + s \text{ is even} \\ 0 \le p,s \le m}} d\cMatA_{p,s}^*d\cMatA_{p,s}$ and
	\begin{align}
		\label{f def}
		f(\cSetMatA) &= -\sum\limits_{\substack{p + s \text{ is even} \\ 0 \le p,s \le m}} \tr \cMatA_{p,s}^*\cMatA_{p,s} + \log h(\cSetMatA); \\
		\notag
		g(\cSetMatA) &= (h(\cSetMatA)^{c_m} + n^{-1/2}\mathtt{p}_a(\cSetMatA) ) \exp\Bigg\{-c_m\sum\limits_{\substack{p + s \text{ is even} \\ 0 \le p,s \le m}} \tr \cMatA_{p,s}^*\cMatA_{p,s} \Bigg\}; \\
		\label{h def}
		h(\cSetMatA) &= \det A + n^{-1/2} \tilde{h}(\cMatA_2) + n^{-1}\mathtt{p}_c(\hat{\cSetMatA}); \\
		\label{A def}
		A &= A(\cMatA_1) = \begin{pmatrix}
			-Z 			& \cMatA_{1} \\
			-\cMatA_{1}^* 	& -Z^*
		\end{pmatrix}
	\end{align}
	with $\mathtt{p}_a(\cSetMatA)$, $\mathtt{p}_c(\hat{\cSetMatA})$ and $\tilde{h}(\cMatA_2)$ being certain polynomials specified in the proof below, and $\hat{\cSetMatA}$ containing all $\cMatA_{p,s}$ except $\cMatA_1$.
\end{prop}
\begin{rem}
	Let $\cMatA_1 = U\Lambda V^*$ be the singular value decomposition of the matrix $\cMatA_1$, i.e.\ $\Lambda = \diag\{\lambda_j\}_{j = 1}^m$, $\lambda_j \ge 0$, $U, V \in U(m)$. In order to perform asymptotic analysis let us change the variables $\cMatA_1 = U\Lambda V^*$ in \eqref{IR result}. Since the Jacobian is $\frac{2^m\pi^{m^2}}{\left(\prod_{j = 1}^{m - 1} j!\right)^2}\Vanddet^2(\Lambda^2) \prod\limits_{j = 1}^m \lambda_j$ (see e.g.\ \cite{Hu:63}) we obtain
	\begin{equation}\label{IR SVD}
	\begin{split}
	\CF_m = Cn^{c_m} &\int\limits_\idom \Vanddet^2(\Lambda^2) \prod\limits_{j = 1}^m \lambda_j \left[ g_0(\Lambda, \hat{\cSetMatA}) + \frac{1}{\sqrt{n}}g_r(U\Lambda V^*, \hat{\cSetMatA}) \right] \\
	&\times \exp\left\{(n - c_m)\left[ f_0(\Lambda, \hat{\cSetMatA}) + \frac{1}{\sqrt{n}}f_r(U\Lambda V^*, \hat{\cSetMatA}) \right]\right\} d\mu(U) d\mu(V) d\Lambda d\hat{\cSetMatA},
	\end{split}
	\end{equation}
	where $\idom = \{(\Lambda, U, V, \hat{\cSetMatA}) \mid \lambda_j \ge 0,\, j = 1, \dotsc, m,\, U, V \in U(m)\}$, $\mu$ is a Haar measure, $d\Lambda = \prod\limits_{j = 1}^m d\lambda_j$ and
	\begin{align}
		\label{f_0 def}
		f_0(\cSetMatA) &= -\sum\limits_{\substack{p + s \text{ is even} \\ 0 \le p,s \le m}} \tr \cMatA_{p,s}^*\cMatA_{p,s} + \log h_0(\cMatA_1); \\
		\notag
		g_0(\cSetMatA) &= h_0(\cMatA_1)^{c_m} \exp\Bigg\{-c_m\sum\limits_{\substack{p + s \text{ is even} \\ 0 \le p,s \le m}} \tr \cMatA_{p,s}^*\cMatA_{p,s} \Bigg\} = e^{c_m f_0(\cSetMatA)}; \\
		\label{h_0 def}
		h_0(\cMatA_1) &= \det \left(A + \frac{1}{\sqrt{n}}\begin{pmatrix}
			\cMatPos	& 0 		\\
			0		& \cMatPos^*
		\end{pmatrix}\right) = \prod\limits_{j = 1}^m (\abs{z_0}^2 + \lambda_j^2); \\
		\label{f_r def}
		f_r(\cSetMatA) &= \sqrt{n}(f(\cSetMatA) - f_0(\cSetMatA)); \\
		\notag 
		g_r(\cSetMatA) &= \sqrt{n}(g(\cSetMatA) - g_0(\cSetMatA)).
	\end{align}
	Notice that $f_0(U\Lambda V^*, \hat{\cSetMatA}) = f_0(\Lambda, \hat{\cSetMatA})$ and the same for $g_0$.
\end{rem}

\begin{rem}\label{rem:case m = 1}
	In the special case $m = 1$ we have
	\begin{equation*} 
		\CF_1(z) = \frac{n}{\pi} \int\exp \left\{n(-\abs{q}^2+\log(\abs{z}^2+\abs{q}^2))\right\} d\bar{q}dq.
	\end{equation*}
	Changing variables to polar coordinates and performing a simple Laplace integration, we obtain
	\begin{equation}\label{F_1 behavior}
	\CF_1(z) = 2n \int\limits_0^{+\infty} r\exp \left\{n(-r^2+\log(\abs{z}^2+r^2))\right\} dr = \sqrt{2\pi n}\, e^{n(\abs{z}^2 - 1)}(1 + o(1)).
	\end{equation}
\end{rem}

\begin{rem}
	In the Gaussian case the representations \eqref{IR result} and \eqref{IR SVD} become much more simple and have the form
	\begin{equation}\label{IR result Gauss}
	\begin{split}
	\CF_m &= \left(\frac{n}{\pi}\right)^{m^2} \int e^{nf(Q_1)} dQ_1^*dQ_1 \\
	&= Cn^{m^2} \int\limits_{\R_+^m} \int\limits_{U(m)} \int\limits_{U(m)} \Vanddet^2(\Lambda^2) \prod\limits_{j = 1}^m \lambda_j \times e^{nf(U\Lambda V^*)} d\mu(U) d\mu(V) d\Lambda,
	\end{split}
	\end{equation}
	where
	\begin{equation}\label{f Gauss}
	f(Q_1) = -\tr Q_1^*Q_1 + \log \det A.
	\end{equation}
	
\end{rem}

\subsection{Proof of Proposition \ref{prop:IR}}

The proof is strongly relied on the SUSY techniques. A reader who is not familiar with Grassmann variables can find all the necessary facts in~\cite{Ef:98} of \cite{Ef:83}. For more serious introduction to SUSY see \cite{Be:87}.

The key formulas of the subsection are well-known Gaussian integration formula
\begin{equation} \label{Gauss int}
\int\limits_{\Compl^n} \exp\left\{-\cVecN^*B\cVecN - \cVecN^*\cVecGI_2 - \cVecGI_1^*\cVecN\right\} d\cVecN^*d\cVecN = 
\pi^n {\det}^{-1} B \exp \{\cVecGI_1^*B^{-1}\cVecGI_2\},
\end{equation}
valid for any positive definite matrix $B$ and even Grassmann variables (i.e.\ sums of products of even number of Grassmann variables) $\cVecGI_1$, $\cVecGI_2$, and its Grassmann analog
\begin{gather}
	\label{Grass int}
	\int \exp\left\{-\aConjMat{\aVecH}B\aVecH - \aConjMat{\aVecH}\aVecGrI_2 - \aConjMat{\aVecGrI_1}\aVecH\right\} d\aConjMat{\aVecH}d\aVecH = 
	\det B \exp \{\aConjMat{\aVecGrI_1}B^{-1}\aVecGrI_2\} 
\end{gather}
valid for arbitrary complex matrix $B$ and odd Grassmann variables (i.e.\ sums of products of odd number of Grassmann variables) $\aConjMat{\aVecGrI_1}$, $\aVecGrI_2$. Rewrite the expression \eqref{F_m} for $\CF_m$ using~\eqref{Grass int} and~\eqref{Gin ens}
\begin{align*}
	\CF_m = \E \Bigg\{
	\int \exp \Bigg\{ -\sum\limits_{j = 1}^{m} \aConjMat{\aVecB_j}\left(\frac{1}{\sqrt{n}}\cMatGin - z_j\right)\aVecB_j - \sum\limits_{j = 1}^{m} \aConjMat{\aVecC_j}\left(\frac{1}{\sqrt{n}}\cMatGin - z_j\right)^*\aVecC_j \Bigg\} d\aMatB d\aMatC
	\Bigg\},
\end{align*}
where $\aVecB_j$, $\aVecC_j$, $j = 1, \dotsc, m$ are $n$-dimensional vectors with components $\aSclB_{kj}$ and $\aSclC_{kj}$ respectively, $d\aMatB = \prod\limits_{j = 1}^{m} d\aVecB_j^+ d\aVecB_j$ and $d\aMatC = \prod\limits_{j = 1}^{m} d\aConjMat{\aVecC_j} d\aVecC_j$. The terms in the exponent can be rearranged as following
\begin{align*}
	-\sum\limits_{j = 1}^m \aConjMat{\aVecB_j}\cMatGin\aVecB_j &= -\tr \aConjMat{\aMatB}\cMatGin\aMatB = \tr \aMatB\aConjMat{\aMatB}\cMatGin = \sum\limits_{k,l = 1}^n (\aMatB\aConjMat{\aMatB})_{lk}x_{kl}, \\
	-\sum\limits_{j = 1}^m \aConjMat{\aVecC_j}\cMatGin^*\aVecC_j &= -\tr \aConjMat{\aMatC}\cMatGin^*\aMatC = \tr \aMatC\aConjMat{\aMatC}\cMatGin^* = \sum\limits_{k,l = 1}^n (\aMatC\aConjMat{\aMatC})_{kl}\bar{x}_{kl}, \\
	\sum\limits_{j = 1}^m \aConjMat{\aVecB_j}z_j\aVecB_j &= \sum\limits_{j = 1}^m \sum\limits_{k = 1}^n \aConjScl{\aSclB_{kj}}z_j\aSclB_{kj} = \sum\limits_{k = 1}^n \sum\limits_{j = 1}^m \aConjScl{\aSclB_{kj}}z_j\aSclB_{kj} = \sum\limits_{k = 1}^n \aConjMat{\aVecBt_k}Z\aVecBt_k, \\
	\sum\limits_{j = 1}^m \aConjMat{\aVecC_j}\bar{z}_j\aVecC_j &= \sum\limits_{j = 1}^m \sum\limits_{k = 1}^n \aSclC_{kj}^*\bar{z}_j\aSclC_{kj} = \sum\limits_{k = 1}^n \sum\limits_{j = 1}^m \aSclC_{kj}^*\bar{z}_j\aSclC_{kj} = \sum\limits_{k = 1}^n \aConjMat{\aVecCt_k}Z^*\aVecCt_k,
\end{align*}
where $\aMatC$ and $\aMatB$ are matrices composed of columns $\aVecC_1, \ldots, \aVecC_m$ and $\aVecB_1, \ldots, \aVecB_m$ respectively, $\aVecBt_k = (\aMatB^T)_k$, $\aVecCt_k = (\aMatC^T)_k$, $Z$~is defined in \eqref{Z def}. Hence
\begin{multline} \label{IR 1}
	\CF_m = \E \Bigg\{
	\int \exp \Bigg\{ \sum\limits_{k = 1}^n \aConjMat{\aVecBt_k}Z\aVecBt_k + \sum\limits_{k = 1}^n \aConjMat{\aVecCt_k}Z^*\aVecCt_k \\
	+ \frac{1}{\sqrt{n}} \sum\limits_{k,l = 1}^n (\aMatB\aMatB^+)_{lk}x_{kl} + \frac{1}{\sqrt{n}} \sum\limits_{k,l = 1}^n (\aMatC\aConjMat{\aMatC})_{kl}\bar{x}_{kl} \Bigg\} d\aMatB d\aMatC
	\Bigg\}.
\end{multline}

To simplify the reading, the remaining steps are first explained in the case when the entries of $\cMatGin$ are Gaussian.

\subsubsection{Gaussian case}

Taking the expectation in \eqref{IR 1} we get
\begin{align*}
	\CF_m = \int& \exp \Bigg\{ \sum\limits_{k = 1}^n \aConjMat{\aVecBt_k}Z\aVecBt_k + \sum\limits_{k = 1}^n \aConjMat{\aVecCt_k}Z^*\aVecCt_k 
	+ \sum\limits_{k,l = 1}^n \frac{1}{n} (\aMatB\aConjMat{\aMatB})_{lk} (\aMatC\aConjMat{\aMatC})_{kl} \Bigg\} d\aMatB d\aMatC.
\end{align*}
Notice that 
\begin{equation*}
	\sum\limits_{k,l = 1}^n (\aMatB\aConjMat{\aMatB})_{lk} (\aMatC\aConjMat{\aMatC})_{kl} = \tr \aMatB\aConjMat{\aMatB}\aMatC\aConjMat{\aMatC} = - \tr \aMatC^T\aConjMat{(\aMatB^T)}\aMatB^T\aConjMat{(\aMatC^T)}.
\end{equation*}
Then the Hubbard--Stra\-to\-no\-vich transformation is applied. The transformation is an application of \eqref{Gauss int} in the reverse direction. It yields
\begin{equation}\label{IR Gauss case 1}
\begin{split}
\CF_m = \left(\frac{n}{\pi}\right)^{m^2} \int \exp \Bigg\{ &\sum\limits_{k = 1}^n \aConjMat{\aVecBt_k}Z\aVecBt_k + \sum\limits_{k = 1}^n \aConjMat{\aVecCt_k}Z^*\aVecCt_k + \tr\aMatC^T\aConjMat{(\aMatB^T)}\cMatA_1 \\
&- \tr \cMatA_1^*\aMatB^T\aConjMat{(\aMatC^T)} - n\tr \cMatA_1^*\cMatA_1 \Bigg\} d\aMatB d\aMatC dQ_1^*dQ_1,
\end{split}
\end{equation}
where $\cMatA_1$ is a $m \times m$ matrix. Transforming the terms
\begin{align*}
	\tr \aMatC^T\aConjMat{(\aMatB^T)}\cMatA_1 &= - \tr \aConjMat{(\aMatB^T)}\cMatA_1\aMatC^T = - \sum\limits_{k = 1}^n \aVecBt_k^+\cMatA_1\aVecCt_k, \\
	\tr \cMatA_1^*\aMatB^T\aConjMat{(\aMatC^T)} &= - \tr \aConjMat{(\aMatC^T)}\cMatA_1^*\aMatB^T = - \sum\limits_{k = 1}^n \aConjMat{\aVecCt_k}\cMatA_1^*\aVecBt_k,
\end{align*}
one can rewrite \eqref{IR Gauss case 1} in the form
\begin{equation*}
	\CF_m = \left(\frac{n}{\pi}\right)^{m^2} \int d\cMatA_1^*d\cMatA_1 e^{- n\tr \cMatA_1^*\cMatA_1} \prod\limits_{k = 1}^{n} \int e^{-\aConjMat{\aVecD_k}A\aVecD_k} d\aVecBt_k^+ d\aVecBt_k d\aConjMat{\aVecCt_k} d\aVecCt_k,
\end{equation*}
where $A$ is defined in \eqref{A def} and
\begin{equation}\label{Psi def}
\aVecD_k = \begin{pmatrix}
\aVecBt_k \\
\aVecCt_k
\end{pmatrix}.
\end{equation}
Finally, integration via \eqref{Grass int} leads us to \eqref{IR result Gauss}.

\subsubsection{General case}

In order to treat the general case let us introduce a notation for a kind of ``Laplace--Fourier transform''
\begin{equation*}
	\charfp{t_1}{t_2} := \E \left\{
	e^{t_1x_{11} + t_2\bar{x}_{11}}
	\right\}.
\end{equation*}
Then the expectation in \eqref{IR 1} can be written in the following form
\begin{align*}
	\CF_m = \int& \prod\limits_{k,l = 1}^n \charfp{\frac{1}{\sqrt{n}} (\aMatB\aConjMat{\aMatB})_{lk}}{\frac{1}{\sqrt{n}} (\aMatC\aConjMat{\aMatC})_{kl}} \\
	&\times \exp \Bigg\{ \sum\limits_{k = 1}^n \aVecBt_k^+Z\aVecBt_k + \sum\limits_{k = 1}^n \aConjMat{\aVecCt_k}Z^*\aVecCt_k \Bigg\} d\aMatB d\aMatC \\
	= \int& \exp \Bigg\{ \sum\limits_{k = 1}^n \aVecBt_k^+Z\aVecBt_k + \sum\limits_{k = 1}^n \aConjMat{\aVecCt_k}Z^*\aVecCt_k \\
	&\phantom{ \exp \Bigg\{} + \sum\limits_{k,l = 1}^n \log\charfp{\frac{1}{\sqrt{n}} (\aMatB\aConjMat{\aMatB})_{lk}}{\frac{1}{\sqrt{n}} (\aMatC\aConjMat{\aMatC})_{kl}} \Bigg\} d\aMatB d\aMatC.
\end{align*}
Expansion of $\log\Phi$ into series gives us
\begin{align}
	\notag
	\CF_m = \int \exp \Bigg\{& \sum\limits_{k = 1}^n \aVecBt_k^+Z\aVecBt_k + \sum\limits_{k = 1}^n \aConjMat{\aVecCt_k}Z^*\aVecCt_k \\
	\label{IR 2}
	&+ \sum\limits_{k,l = 1}^n \sum\limits_{p,s = 0}^m \frac{\cumul{p}{s}}{p!s!}  \frac{1}{n^{(p + s)/2}} \left((\aMatB\aConjMat{\aMatB})_{lk}\right)^p\left((\aMatC\aConjMat{\aMatC})_{kl}\right)^s \Bigg\} d\aMatB d\aMatC,
\end{align}
with
\begin{align*} 
	\cumul{p}{s} = \left.\frac{\partial^{p + s}}{\partial^p t_1 \partial^s t_2} \log\charfp{t_1}{t_2}\right|_{t_1 = t_2 = 0}.
\end{align*}
In particular, 
\begin{equation}\label{cumul values}
\begin{split}
\cumul{0}{0} &= 0; \\
\cumul{1}{0} &= \conj{\cumul{0}{1}} = \E\{x_{11}\} = 0; \\
\cumul{2}{0} &= \conj{\cumul{0}{2}} = \E\{x_{11}^2\} - \E^2\{x_{11}\} = 0; \\
\cumul{1}{1} &= \E\{\abs{x_{11}}^2\} - \abs{\E\{x_{11}\}}^2 = 1.
\end{split}
\end{equation}
Let us transform the terms in the exponent again
\begin{align}
	\notag
	\sum\limits_{k,l = 1}^n &\left((\aMatB\aConjMat{\aMatB})_{lk}\right)^p\left((\aMatC\aConjMat{\aMatC})_{kl}\right)^s \\
	\notag
	&= \sum\limits_{k,l = 1}^n \Bigg(\sum\limits_{j = 1}^m \aSclB_{lj}^{\phantom{+}}\aConjScl{\aSclB_{kj}}\Bigg)^p\Bigg(\sum\limits_{j = 1}^m \aSclC_{kj}^{\phantom{+}}\aConjScl{\aSclC_{lj}}\Bigg)^s = p!s! \sum\limits_{k,l = 1}^n \sum\limits_{\substack{\alpha \in \indexset_{m, p} \\ \beta \in \indexset_{m, s}}} \prod\limits_{q = 1}^{p} \aSclB_{l\alpha_{q}}^{\phantom{+}}\aConjScl{\aSclB_{k\alpha_{q}}} \prod\limits_{r = 1}^{s}\aSclC_{k\beta_{r}}^{\phantom{+}}\aConjScl{\aSclC_{l\beta_{r}}} \\
	\notag
	&= (-1)^{p^2} p!s! \sum\limits_{k,l = 1}^n \sum\limits_{\substack{\alpha \in \indexset_{m, p} \\ \beta \in \indexset_{m, s}}} \prod\limits_{r = s}^{1} \aSclC_{k\beta_{r}}^{\phantom{+}} \prod\limits_{q = p}^{1} \aConjScl{\aSclB_{k\alpha_{q}}} \prod\limits_{q = 1}^{p} \aSclB_{l\alpha_{q}}^{\phantom{+}} \prod\limits_{r = 1}^{s} \aConjScl{\aSclC_{l\beta_{r}}} \\
	\label{sum order change}
	&= p!s! \sum\limits_{\substack{\alpha \in \indexset_{m, p} \\ \beta \in \indexset_{m, s}}} \Bigg(\sum\limits_{k = 1}^n (-1)^p\prod\limits_{r = s}^{1} \aSclC_{k\beta_{r}}^{\phantom{+}} \prod\limits_{q = p}^{1} \aConjScl{\aSclB_{k\alpha_{q}}}\Bigg) \Bigg(\sum\limits_{k = 1}^n \prod\limits_{q = 1}^{p} \aSclB_{k\alpha_{q}}^{\phantom{+}} \prod\limits_{r = 1}^{s} \aConjScl{\aSclC_{k\beta_{r}}}\Bigg),
\end{align}
where 
\begin{equation} \label{indexset def}
\indexset_{m, p'} = \{\alpha \in \Z^{p'} \mid 1 \leq \alpha_1 < \ldots < \alpha_{p'} \leq m\}
\end{equation}
At this point the Hubbard--Stra\-to\-no\-vich transformation is applied. As it was mentioned before, the transformation is an employment of \eqref{Gauss int} or \eqref{Grass int} in the reverse direction. It yields for even $p + s$
\begin{align}
	\notag
	\exp \Bigg\{& \cumul{p}{s} n^{-(p + s)/2} \Bigg(\sum\limits_{k = 1}^n (-1)^p\prod\limits_{r = s}^{1} \aSclC_{k\beta_{r}}^{\phantom{+}} \prod\limits_{q = p}^{1} \aConjScl{\aSclB_{k\alpha_{q}}}\Bigg) \Bigg(\sum\limits_{k = 1}^n \prod\limits_{q = 1}^{p} \aSclB_{k\alpha_{q}} \prod\limits_{r = 1}^{s} \aConjScl{\aSclC_{k\beta_{r}}}\Bigg) \Bigg\} \\
	\notag
	&= \frac{n}{\pi} \int \exp \Bigg\{
	- n^{-\frac{p + s - 2}{4}} \sum\limits_{k = 1}^n \tilde{\sclA}_{\beta\alpha}^{(k,p,s)} \cSclA_{\alpha\beta}^{(p,s)} - n^{-\frac{p + s - 2}{4}} \sum\limits_{k = 1}^n \cConjScl{\cSclA}_{\alpha\beta}^{(p,s)} \sclA_{\alpha\beta}^{(k,p,s)} - n\abs{\cSclA_{\alpha\beta}^{(p,s)}}^2 \Bigg\} \\
	\label{HS commuting}
	&\times d\cConjScl{\cSclA}_{\alpha\beta}^{(p,s)}d\cSclA_{\alpha\beta}^{(p,s)},
\end{align}
where
\begin{equation}\label{chi def}
\begin{split}
\tilde{\sclA}_{\beta\alpha}^{(k,p,s)} &= \sqrt{\cumul{p}{s}}(-1)^p\prod\limits_{r = s}^{1} \aSclC_{k\beta_{r}}^{\phantom{+}} \prod\limits_{q = p}^{1} \aConjScl{\aSclB_{k\alpha_{q}}}; \\
\sclA_{\alpha\beta}^{(k,p,s)}&= \sqrt{\cumul{p}{s}} \prod\limits_{q = 1}^{p} \aSclB_{k\alpha_{q}} \prod\limits_{r = 1}^{s} \aConjScl{\aSclC_{k\beta_{r}}}.
\end{split}
\end{equation}
Here and below we take a branch of the square root such that its argument is in~$\left[ 0, \pi \right)$. Similarly, for odd $p + s$ we have
\begin{align}
	\notag
	\exp \Bigg\{ \cumul{p}{s} n^{-(p + s)/2}& \Bigg(\sum\limits_{k = 1}^n (-1)^p\prod\limits_{r = 1}^{s} \aSclC_{k\beta_{r}}^{\phantom{+}} \prod\limits_{q = 1}^{p} \aConjScl{\aSclB_{k\alpha_{q}}}\Bigg) \Bigg(\sum\limits_{k = 1}^n \prod\limits_{q = 1}^{p} \aSclB_{k\alpha_{q}} \prod\limits_{r = 1}^{s} \aConjScl{\aSclC_{k\beta_{r}}}\Bigg) \Bigg\} \\
	\notag
	= \int \exp \Bigg\{
	&- n^{-\frac{p + s}{4}} \sum\limits_{k = 1}^n \tilde{\sclA}_{\beta\alpha}^{(k,p,s)} \aSclA_{\alpha\beta}^{(p,s)} - n^{-\frac{p + s}{4}} \sum\limits_{k = 1}^n \aConjScl{\left(\aSclA_{\alpha\beta}^{(p,s)}\right)} \sclA_{\alpha\beta}^{(k,p,s)} \\
	\label{HS anticommuting}
	&- \aConjScl{\left(\aSclA_{\alpha\beta}^{(p,s)}\right)}\aSclA_{\alpha\beta}^{(p,s)} \Bigg\} d\aConjScl{\left(\aSclA_{\alpha\beta}^{(p,s)}\right)}d\aSclA_{\alpha\beta}^{(p,s)}.
\end{align}

Then the combination of \eqref{IR 2}, \eqref{sum order change}, \eqref{HS commuting} and \eqref{HS anticommuting} gives us
\begin{align}
	\CF_m = \left(\frac{n}{\pi}\right)^{c_m} \int& \prod\limits_{k = 1}^n \mathsf{j}_k \prod\limits_{\substack{p + s \text{ is odd} \\ 0 \le p,s \le m}} e^{-\tr\aConjMat{\aMatA_{p,s}}\aMatA_{p,s}} d\aConjMat{\aMatA_{p,s}} d\aMatA_{p,s} \prod\limits_{\substack{p + s \text{ is even} \\ 0 \le p,s \le m}} e^{-n\tr \cMatA_{p,s}^*\cMatA_{p,s}}d\cMatA_{p,s}^*d\cMatA_{p,s}
	\label{after HS}
\end{align}
where 
\begin{align}
	\label{j_k def}
	\begin{split}
		\mathsf{j}_k &= \int \exp \left\{ b_{k,2} + n^{-1/2}b_{k,4} + n^{-3/4}\mathtt{p}_a^{(1)}(\aSetMatA, \aMatB, \aMatC) + n^{-1}\mathtt{p}_c^{(1)}(\hat{\cSetMatA}, \aMatB, \aMatC) \right\} \\
		&\phantom{= \int} \times 
		d\aConjMat{\aVecBt_k} d\aVecBt_k d\aConjMat{\aVecCt_k} d\aVecCt_k,
	\end{split} \\
	\notag
	b_{k,2} &= - \left(\tr \tilde{\matA}_{k,1,1} \cMatA_{1,1} + \tr \cMatA_{1,1}^* \matA_{k,1,1}\right) + \aVecBt_k^+Z\aVecBt_k + \aConjMat{\aVecCt_k}Z^*\aVecCt_k, \\
	\label{S_k,4 def}
	b_{k,4} &= - \sum\limits_{p + s = 4} \left(\tr \tilde{\matA}_{k,p,s} \cMatA_{p,s} + \tr \cMatA_{p,s}^* \matA_{k,p,s}\right), \\
	\notag
	\mathtt{p}_a^{(1)}(\aSetMatA, \aMatB, \aMatC) &= - \sum\limits_{j = 2}^m n^{-(j - 2)/2} \sum\limits_{p + s = 2j - 1} \left(\tr \tilde{\matA}_{k,p,s} \aMatA_{p,s} + \tr \aConjMat{\aMatA_{p,s}} \matA_{k,p,s}\right), \\
	\notag
	\mathtt{p}_c^{(1)}(\hat{\cSetMatA}, \aMatB, \aMatC) &= - \sum\limits_{j = 3}^m n^{-(j - 3)/2} \sum\limits_{p + s = 2j} \left(\tr \tilde{\matA}_{k,p,s} \cMatA_{p,s} + \tr \cMatA_{p,s}^* \matA_{k,p,s}\right).
\end{align}
In the formulas above $\aMatA_{p,s}$, 
$\cMatA_{p,s}$, $\tilde{\matA}_{k,p,s}$ and $\matA_{k,p,s}$ are matrices whose entries are $\aSclA_{\alpha\beta}^{(p,s)}$, 
$\cSclA_{\alpha\beta}^{(p,s)}$, $\tilde{\sclA}_{\beta\alpha}^{(k,p,s)}$ and $\sclA_{\alpha\beta}^{(k,p,s)}$ respectively. The rows and columns are indexed by elements of the set $\indexset_{m,p}$ for corresponding $p$ (or $s$) in the lexicographical order. Note also that $\mathtt{p}_a^{(1)}$ and $\mathtt{p}_c^{(1)}$ are the first degree homogeneous polynomials of the entries of $\aSetMatA$ and $\hat{\cSetMatA}$ respectively, where $\hat{\cSetMatA}$ contains all the $\cMatA_{p,s}$ except $\cMatA_1$. One more thing we need is that all the monomials of $\mathtt{p}_a^{(1)}$ have odd degree w.r.t.\ $\aVecBt_k$ and $\aVecCt_k$, and all the monomials of $\mathtt{p}_c^{(1)}$ have even degree w.r.t.\ $\aVecBt_k$ and $\aVecCt_k$.

Fortunately, the integral in \eqref{after HS} over $\aMatB$ and $\aMatC$ factorizes. Therefore the integration can be performed over $\aVecBt_k$ and $\aVecCt_k$ separately for every $k$. Lemma \ref{lem:int over psi,tau} provides a corresponding result.
\begin{lem}\label{lem:int over psi,tau}
	Let $\mathsf{j}_k$ be defined by \eqref{j_k def}. Then
	\begin{equation}\label{j_k value}
	\mathsf{j}_k = \det A + n^{-1/2} \tilde{h}(\cMatA_2) + n^{-1}\mathtt{p}_c(\hat{\cSetMatA}) + n^{-3/2}\mathtt{p}_{a}^{(2)}(\aSetMatA, \cSetMatA),
	\end{equation}
	where A is defined in \eqref{A def},
	\begin{equation}\label{tilde h def}
	\tilde{h}(\cMatA_2) = -\int \left(\tr \tilde{\matA}_{k,2,2} \cMatA_{2} + \tr \cMatA_{2}^* \matA_{k,2,2}\right) e^{b_{k,2}} d\aConjMat{\aVecBt_k} d\aVecBt_k d\aConjMat{\aVecCt_k} d\aVecCt_k,
	\end{equation}
	$\mathtt{p}_c(\hat{\cSetMatA})$ and $\mathtt{p}_{a}^{(2)}(\aSetMatA, \cSetMatA)$ are polynomials such that 
	\begin{enumerate}[label = (\roman*)]
		\item\label{P_c} $\mathtt{p}_c(0) = 0$;
		\item\label{P_a} every monomial of $\mathtt{p}_a^{(2)}$ has at least second degree w.r.t.\ $\aSetMatA$.
	\end{enumerate}
\end{lem}
\begin{proof}
	The integral $\mathsf{j}_k$ is computed by the expansion of the exponent into series. We start with
	\begin{align}
		\notag
		\mathsf{j}_k = \int{}& \Bigg(1 + \sum\limits_{1 \le k \le 4m/3} n^{-3k/4}(\mathtt{p}_a^{(1)}(\aSetMatA, \aMatB, \aMatC))^k \Bigg) e^{ b_{k,2} + n^{-1/2}b_{k,4} + n^{-1}\mathtt{p}_c^{(1)}(\hat{\cSetMatA}, \aMatB, \aMatC)} \\
		\label{psi,tau integration}
		&\times d\aVecBt_k^+ d\aVecBt_k d\aConjMat{\aVecCt_k} d\aVecCt_k,
	\end{align}
	where the terms of degree higher than $4m$ w.r.t.\ $\aVecBt_k$ and $\aVecCt_k$ vanish, because the square of any anti-commuting variable is zero. The monomials of odd degree w.r.t.\ $\aVecBt_k$ and $\aVecCt_k$ also vanish after integration. Indeed, for every odd degree homogeneous polynomial $\tilde{\mathtt{p}}$ the expansion of $\tilde{\mathtt{p}}\left(\aVecBt_k, \aVecCt_k\right)e^{b_{k,2} + n^{-1/2}b_{k,4} + n^{-1}\mathtt{p}_c(\hat{\cSetMatA}, \aMatB, \aMatC)}$ into series gives us only odd degree terms. Whereas the number of Grassmann variables is even, there are no top degree monomials and the integral is zero. Thus \eqref{psi,tau integration} simplifies to
	\begin{equation}\label{psi,tau integration 2}
	\begin{split}
	\mathsf{j}_k = \int{}& \left(1 + n^{-3/2}\mathtt{p}_a^{(3)}(\aSetMatA, \aMatB, \aMatC)\right)e^{ b_{k,2} + n^{-1/2}b_{k,4} + n^{-1}\mathtt{p}_c^{(1)}(\hat{\cSetMatA}, \aMatB, \aMatC)} \\
	&\times d\aVecBt_k^+ d\aVecBt_k d\aConjMat{\aVecCt_k} d\aVecCt_k,
	\end{split}
	\end{equation}
	where $\mathtt{p}_a^{(3)}(\aSetMatA, \aMatB, \aMatC)$ is a polynomial and its every monomial has degree at least 2 w.r.t.\ $\aSetMatA$ and at least 2 w.r.t.\ $\aVecBt_k$ and $\aVecCt_k$. Note that
	\begin{equation}\label{negligible int}
	\int \mathtt{p}_a^{(3)}(\aSetMatA, \aMatB, \aMatC)e^{ b_{k,2} + n^{-1/2}b_{k,4} + n^{-1}\mathtt{p}_c^{(1)}(\hat{\cSetMatA}, \aMatB, \aMatC)} d\aVecBt_k^+ d\aVecBt_k d\aConjMat{\aVecCt_k} d\aVecCt_k = \mathtt{p}_a^{(2)}(\aSetMatA, \cSetMatA),
	\end{equation}
	where $\mathtt{p}_a^{(2)}(\aSetMatA, \cSetMatA)$ satisfies condition \ref{P_a}. Substitution of \eqref{negligible int} into \eqref{psi,tau integration 2} yields
	\begin{equation*} 
		\mathsf{j}_k = \int e^{ b_{k,2} + n^{-1/2}b_{k,4} + n^{-1}\mathtt{p}_c^{(1)}(\hat{\cSetMatA}, \aMatB, \aMatC)} d\aVecBt_k^+ d\aVecBt_k d\aConjMat{\aVecCt_k} d\aVecCt_k + n^{-3/2}\mathtt{p}_a^{(2)}(\aSetMatA, \cSetMatA).
	\end{equation*}
	Further expansion implies
	\begin{equation*} 
		\begin{split}
			\mathsf{j}_k ={} &\int \left(1 + n^{-1/2}b_{k,4} + n^{-1}\mathtt{p}_c^{(2)}(\hat{\cSetMatA}, \aMatB, \aMatC)\right)e^{ b_{k,2} } d\aVecBt_k^+ d\aVecBt_k d\aConjMat{\aVecCt_k} d\aVecCt_k \\
			&+ n^{-3/2}\mathtt{p}_a^{(2)}(\aSetMatA, \cSetMatA),
		\end{split}
	\end{equation*}
	where $\mathtt{p}_c^{(2)}(\hat{\cSetMatA}, \aMatB, \aMatC)$ is again a polynomial such that $\mathtt{p}_c^{(2)}(0, \aMatB, \aMatC) = 0$. Similarly to above we obtain
	\begin{equation*} 
		\begin{split}
			\mathsf{j}_k ={} &\int \left(1 + n^{-1/2}b_{k,4}\right)e^{ b_{k,2} } d\aVecBt_k^+ d\aVecBt_k d\aConjMat{\aVecCt_k} d\aVecCt_k \\
			&+ n^{-1}\mathtt{p}_c(\hat{\cSetMatA}) + n^{-3/2}\mathtt{p}_a^{(2)}(\aSetMatA, \cSetMatA),
		\end{split}
	\end{equation*}
	where $\mathtt{p}_c(\hat{\cSetMatA})$ satisfies condition \ref{P_c}.
	
	Let us consider the expression \eqref{S_k,4 def} for $b_{k,4}$ in more detail. Every term in \eqref{S_k,4 def} with $(p,s) \ne (2,2)$ has different numbers of ``non-conjugate'' Grassmann variables (without ``$+$'' superscript) and ``conjugates'' (with ``$+$'' superscript). But every term of $b_{k,2}$ has equal number of ``non-conjugate'' and ``conjugate'' Grassmann variables. The same is true for the expansion of $e^{b_{k,2}}$ and for top degree monomial of $\aVecBt_k$ and $\aVecCt_k$. Hence for $(p,s) \ne (2,2)$, $p + s = 4$
	\begin{equation*}
		\int \left(\tr \tilde{\matA}_{k,p,s} Q_{p,s} + \tr Q_{p,s}^* \matA_{k,p,s}\right)e^{ b_{k,2} } d\aVecBt_k^+ d\aVecBt_k d\aConjMat{\aVecCt_k} d\aVecCt_k = 0.
	\end{equation*}
	Therefore
	\begin{equation} \label{psi,tau integration 4}
	\begin{split}
	\mathsf{j}_k ={} &\int \left(1 
	- n^{-1/2}\left(\tr \tilde{\matA}_{k,2,2} Q_{2} + \tr Q_{2}^* \matA_{k,2,2}\right)\right)e^{ b_{k,2} } d\aVecBt_k^+ d\aVecBt_k d\aConjMat{\aVecCt_k} d\aVecCt_k \\
	&+ n^{-1}\mathtt{p}_c(\hat{\cSetMatA}) + n^{-3/2}\mathtt{p}_a^{(2)}(\aSetMatA, \cSetMatA).
	\end{split}
	\end{equation}
	
	Recalling the definition of $\sclA_{\alpha\beta}^{(k,p,s)}$ \eqref{chi def} and the values of $\cumul{p}{s}$ \eqref{cumul values}, one can render $b_{k,2}$ in the form
	\begin{equation}\label{new form of b_2}
	b_{k,2} = - \aConjMat{\aVecD_k}A\aVecD_k,
	\end{equation}
	where $A$ is defined in \eqref{A def} and $\aVecD_k$ is defined in \eqref{Psi def}. Then \eqref{psi,tau integration 4} and \eqref{Grass int} imply the assertion of the lemma.
\end{proof}
A substitution of \eqref{j_k value} into \eqref{after HS} gives us
\begin{align*}
	\CF_m = \left(\frac{n}{\pi}\right)^{c_m} \int{}& (h(\cSetMatA) + n^{-3/2}\mathtt{p}_{a}^{(2)}(\aSetMatA, \cSetMatA))^n \prod\limits_{\substack{p + s \text{ is odd} \\ 0 \le p,s \le m}} e^{-\tr\aConjMat{\aMatA_{p,s}}\aMatA_{p,s}} d\aConjMat{\aMatA_{p,s}} d\aMatA_{p,s} \\
	&\times \prod\limits_{\substack{p + s \text{ is even} \\ 0 \le p,s \le m}} e^{-n\tr Q_{p,s}^*Q_{p,s}}dQ_{p,s}^*dQ_{p,s},
\end{align*}
where $h(\cSetMatA)$ is defined in \eqref{h def}. Further
\begin{equation*}
	(h(\cSetMatA) + n^{-3/2}\mathtt{p}_{a}^{(2)}(\aSetMatA, \cSetMatA))^n = \sum\limits_{k = 0}^{c_m} \binom{n}{k}n^{-3k/2}h(\cSetMatA)^{n - k}(\mathtt{p}_{a}^{(2)}(\aSetMatA, \cSetMatA))^k
\end{equation*}
because there are $2c_m$ anti-commuting variables and every monomial of $\mathtt{p}_a^{(2)}$ has at least second degree w.r.t.\ $\aSetMatA$. Hence,
\begin{equation}\label{before int over rho}
\begin{split}
	\CF_m = \left(\frac{n}{\pi}\right)^{c_m} \int{}& (h(\cSetMatA)^{c_m} + n^{-1/2}\mathtt{p}_a^{(3)}(\aSetMatA, \cSetMatA)) \prod\limits_{\substack{p + s \text{ is odd} \\ 0 \le p,s \le m}} e^{-\tr\aConjMat{\aMatA_{p,s}}\aMatA_{p,s}} d\aConjMat{\aMatA_{p,s}} d\aMatA_{p,s} \\
	&\times e^{nf(\cSetMatA)-c_m\log h(\cSetMatA)}d\cSetMatA,
\end{split}
\end{equation}
where $\mathtt{p}_a^{(3)}$ is a polynomial and $f(\cSetMatA)$ is defined in \eqref{f def}. Taking into account \eqref{Grass int} and the definition of an integral over anti-commuting variables, one can perform the integration over $\aSetMatA$ in \eqref{before int over rho} and obtain \eqref{IR result}.

\section{Asymptotic analysis}\label{sec:asympt analysis}
The goal of the section is to investigate the asymptotic behavior of the integral representation~\eqref{IR SVD}. To this end, the steepest descent method is applied. As usual, the hardest step is to choose stationary points of $f(\cSetMatA)$ and a $N$-dimensional (real) manifold $M_* \subset \Compl^{N}$ such that for any chosen stationary point $\cSetMatA_* \in M_*$
\begin{equation*}
	\Re f(\cSetMatA) < \Re f(\cSetMatA_*), \quad \forall \cSetMatA \in M_*,\, \text{$\cSetMatA$ is not chosen}.
\end{equation*}
Note that $N$ is equal to the number of real variables of the integration, i.e.\ in our case $N = 2^{2m}$.

The present proof proceeds a slightly different but rather standard scheme for the case when function $f(\cSetMatA)$ has the form
\begin{equation*}
	f(\cSetMatA) = f_0(\cSetMatA) + n^{-1/2}f_r(\cSetMatA),
\end{equation*}
where $f_0(\cSetMatA)$ does not depend on $n$, whereas $f_r(\cSetMatA)$ may depend on $n$. We choose stationary points of $f_0(\cSetMatA)$ of the form $Q_1 = U\lambda_0V^* 
$, 
$\hat{\cSetMatA} = 0$, where $\lambda_0$ is a fixed real number and $U$ and $V$ vary in $U(m)$. Then the steepest descent method is applied to the integral over $\Lambda$ and $\hat{\cSetMatA}$. In the process $U$ and $V$ are being considered as parameters and all the estimates are uniform in $U$ and $V$. As  soon as the domain of integration is restricted by a small neighborhood we recall about the integration over $U$ and $V$. After several changes of the variables the integral is finally computed.

We start with the analysis of $f_0$.
\begin{lem}\label{lem:max of f_0}
	Let the function $f_0 \colon \R^{2^{2m}} \to [-\infty, +\infty)$ be defined by \eqref{f_0 def}. Then $f_0(\Lambda, \hat{\cSetMatA})$ attains its global maximum value only at the point
	\begin{equation*}
		\lambda_1 = \dotsb = \lambda_m = \lambda_0, \quad \hat{\cSetMatA} = 0,
	\end{equation*}
	where $\lambda_0 = \sqrt{1 - \abs{z_0}^2}$. Moreover, the matrix of second order derivatives of $f_0$ w.r.t. $\Lambda$ and $\hat{\cSetMatA}$ at this point is negative definite.
\end{lem}
\begin{proof}
	It is evident from \eqref{f_0 def} and \eqref{h_0 def} that $f_0(\Lambda, \hat{\cSetMatA})$ has the form
	\begin{equation}\label{f_0: variables separation}
	f_0(\Lambda, \hat{\cSetMatA}) = \sum\limits_{j = 1}^m f_*(\lambda_j) - \sum\limits_{(p,s) \ne (1,1)} \tr Q_{p,s}^*Q_{p,s}, 
	\end{equation}
	where
	\begin{equation*} 
		f_*(\lambda) = -\lambda^2 + \log (\abs{z_0}^2 + \lambda^2).
	\end{equation*}
	Since $f_*'(\lambda) = 0$ iff $\lambda = \lambda_0$ and $\lim\limits_{\lambda \to \infty} f_*(\lambda) = -\infty$, $f_*(\lambda)$ attains its global maximum value only at $\lambda = \lambda_0$. Furthermore, $f_*''(\lambda_0) = -4\lambda_0^2$. These facts and \eqref{f_0: variables separation} immediately imply the assertion of the lemma.
\end{proof}

As in the previous section we consider first the Gaussian case and then the general case.

\subsection{Gaussian case}
Now we proceed to the integral estimates. In a standard way the integration domain in \eqref{IR result Gauss} can be restricted as follows
\begin{equation*}
	\CF_m = Cn^{m^2} \int\limits_{\Sigma_r} \Vanddet^2(\Lambda^2) \prod\limits_{j = 1}^m \lambda_j \times e^{nf(U\Lambda V^*)} d\mu(U) d\mu(V) d\Lambda + O(e^{-nr/2}),
\end{equation*}
where
\begin{equation*}
	\Sigma_r = \left\{(\Lambda, U, V) \mid \norm{\Lambda} \le r\right\}.
\end{equation*}
The next step is to restrict the integration domain by 
\begin{equation}\label{stpoinnbh def}
\stpointsnbh = \left\{(\Lambda, U, V) \mid \norm{\Lambda - \Lambda_0} \le \frac{\log n}{\sqrt{n}}\right\},
\end{equation}
where $\Lambda_0 = \lambda_0I$, $I$ is a unit matrix. To this end we need the estimate of $\Re f$ given by the following lemmas.

\begin{lem}\label{lem:f(UL V^*) expansion}
	Let $\tilde{\Lambda}$ be a $m \times m$ diagonal matrix such that $\normsized{\tilde{\Lambda}} \le \log n$. Then uniformly in $U$ and $V$
	\begin{equation} \label{f expansion}
	\begin{split}
	f(U(\Lambda_0 + n^{-1/2}\tilde{\Lambda})V^*) ={}& {- m\lambda_0^2} + n^{-1/2} \tr (\bar{z}_0\cMatPos + z_0\cMatPos^*) + n^{-1} \tr\cMatPos_U\cMatPos_V^* \\
	&- n^{-1} \tr(2\lambda_0\tilde{\Lambda} + \bar{z}_0\cMatPos_U + z_0\cMatPos_V^*)^2/2 + 
	O\big(n^{-3/2}\log^3 n\big)
	\end{split}
	\end{equation}
	where
	\begin{equation}\label{zeta_B def}
	\cMatPos_B = B^*\cMatPos B.
	\end{equation}
\end{lem}

\begin{proof}
	If $Q_1 = U(\Lambda_0 + n^{-1/2}\tilde{\Lambda})V^*$ then $A$ has the form
	\begin{equation*}
		A = \begin{pmatrix}
			U 	& 0 \\
			0 	& V
		\end{pmatrix}\left(A_0 + \frac{1}{\sqrt{n}}A_1\right)\begin{pmatrix}
		U^*	& 0 \\
		0	& V^*
	\end{pmatrix},
\end{equation*}
where
\begin{equation}\label{A_0,A_1 def}
A_0 = \begin{pmatrix}
-Z_0 		& \Lambda_0 \\
-\Lambda_0 	& -Z_0^*
\end{pmatrix}, \quad
A_1 = \begin{pmatrix}
-\cMatPos_U 		& \tilde{\Lambda} \\
-\tilde{\Lambda} 	& -\cMatPos_V^*
\end{pmatrix}.
\end{equation}
Taking into account that
\begin{equation*}
	\det A_0 = \left[\det 
	\begin{pmatrix}
		-z_0		& \lambda_0 \\
		-\lambda_0	& -\cConjScl{z}_0
	\end{pmatrix}\right]^m
	= 1,
\end{equation*}
one gets
\begin{equation}\label{log det A}
\begin{split}
\log \det A &= 
\log \det A_0^{-1}A = \tr \log (1 + n^{-1/2}A_0^{-1}A_1) \\
&= \frac{1}{\sqrt{n}}\tr A_0^{-1}A_1 - \frac{1}{2n}\tr (A_0^{-1}A_1)^2 + O\left(\frac{\log^3 n}{\sqrt{n^3}}\right)
\end{split}
\end{equation}
uniformly in $U$ and $V$. Moreover,
\begin{equation}\label{A_0^-1A_1}
A_0^{-1}A_1 = \begin{pmatrix}
\bar{z}_0\cMatPos_U + \lambda_0\tilde{\Lambda}	& -\bar{z}_0\tilde{\Lambda} + \lambda_0\cMatPos_V^* \\
-\lambda_0\cMatPos_U + z_0\tilde{\Lambda}		& \lambda_0\tilde{\Lambda} + z_0\cMatPos_V^*
\end{pmatrix}.
\end{equation}
Combining \eqref{log det A}, \eqref{A_0^-1A_1} and \eqref{f Gauss}, we get
\begin{equation*}
	\begin{split}
		f(U(\Lambda_0 + n^{-1/2}\tilde{\Lambda})V^*) = \tr\Big[&- \Lambda_0^2 - 2n^{-1/2}\lambda_0\tilde{\Lambda} - n^{-1}\tilde{\Lambda}^2 + n^{-1/2}(2\lambda_0 \tilde{\Lambda} + \bar{z}_0 \cMatPos_U + z_0 \cMatPos_V^*) \\
		&- n^{-1}\big\{(\lambda_0^2 - \abs{z_0}^2) \tilde{\Lambda}^2 + 2\bar{z}_0\lambda_0\cMatPos_U\tilde{\Lambda} + 2z_0\lambda_0\cMatPos_V^*\tilde{\Lambda} \\
		&+ \frac{1}{2}(\bar{z}_0\cMatPos_U + z_0\cMatPos_V^*)^2 - \cMatPos_U\cMatPos_V^*\big\}\Big] + O\big(n^{-3/2}\log^3 n\big).
	\end{split}
\end{equation*}
The last expansion yields \eqref{f expansion}.
\end{proof}

\begin{lem}\label{lem:est for Re f}
	Let $\tilde{f}(Q_1) = f(Q_1) - f(\Lambda_0)$. Then for sufficiently large $n$
	\begin{equation*} 
		\max_{\frac{\log n}{\sqrt{n}} \le \norm{\Lambda - \Lambda_0} \le r} \Re \tilde{f}(U\Lambda V^*) \le -C\frac{\log^2 n}{n}
	\end{equation*}
	uniformly in $U$ and $V$.
\end{lem}
\begin{proof}
	First let us check that the first and the second derivatives of $f_r$ are bounded in the $\delta$-neighborhood of $\Lambda_0$, where $f_r$ is defined in \eqref{f_r def} and $\delta$ is $n$-independent. Indeed, since $h$ and $h_0$ are polynomials and $h \rightrightarrows h_0$ on compacts
	\begin{equation*}
		\abs{\frac{1}{\sqrt{n}}\frac{\partial \Re f_r}{\partial \lambda_j}} \le \abs{\frac{1}{\sqrt{n}}\frac{\partial f_r}{\partial \lambda_j}} = \abs{\frac{\partial (f - f_0)}{\partial \lambda_j}} = \abs{\frac{\partial (\log h - \log h_0)}{\partial \lambda_j}} \le \abs{\frac{1}{h_0} \cdot \frac{\partial h_0}{\partial \lambda_j} - \frac{1}{h} \cdot \frac{\partial h}{\partial \lambda_j}} \le \frac{C}{\sqrt{n}}.
	\end{equation*}
	For every diagonal matrix $E = \diag\{e_j\}$ let $v(E)$ denote a vector with components $e_j$. Then for every diagonal matrix $E$ of unit norm and for $\frac{\log n}{\sqrt{n}} \le t \le \delta$ we have
	\begin{equation*}
		\begin{split}
			\der{}{t} \Re\tilde{f}(U(\Lambda_0 + tE)V^*) ={} &\langle \nabla_\Lambda f_0(U(\Lambda_0 + tE)V^*), v(E) \rangle \\
			&+ n^{-1/2} \langle \nabla_\Lambda \Re f_r(U(\Lambda_0 + tE)V^*), v(E) \rangle \\
			={} &\langle \nabla_\Lambda f_0(\Lambda_0 + tE), v(E) \rangle + O(n^{-1/2}),
		\end{split}
	\end{equation*}
	where $\langle \cdot, \cdot \rangle$ is a standard scalar product. Expanding the scalar product by the Taylor formula and considering that $\nabla_\Lambda f_0(\Lambda_0) = 0$, we obtain
	\begin{equation*}
		\begin{split}
			\der{}{t} \Re\tilde{f}(U(\Lambda_0 + tE)V^*) &= t\langle f_0''(\Lambda_0)v(E), v(E) \rangle + r_1 + O(n^{-1/2}),
		\end{split}
	\end{equation*}
	where $f_0''$ is a matrix of second order derivatives of $f_0$ w.r.t.\ $\Lambda$ and $\abs{r_1} \le Ct^2$. $f_0''(\Lambda_0)$ is negative definite according to Lemma \ref{lem:max of f_0}. Hence $\der{}{t} \Re\tilde{f}(U(\Lambda_0 + tE)V^*)$ is negative and
	\begin{equation}\label{Re f: nbh est}
	\max_{\frac{\log n}{\sqrt{n}} \le \norm{\Lambda - \Lambda_0} \le \delta} \Re \tilde{f}(U\Lambda V^*) = \max_{\norm{\Lambda - \Lambda_0} = \frac{\log n}{\sqrt{n}}} \Re \tilde{f}(U\Lambda V^*) \le f(U\Lambda_0 V^*) - C\frac{\log^2 n}{n} - f(\Lambda_0).
	\end{equation}
	Notice that $f_r$ is bounded from above uniformly in $n$. This fact and Lemma \ref{lem:max of f_0} imply that $\delta$ in \eqref{Re f: nbh est} can be replaced by $r$
	\begin{equation*}
		\max_{\frac{\log n}{\sqrt{n}} \le \norm{\Lambda - \Lambda_0} \le r} \Re \tilde{f}(U\Lambda V^*) \le f(U\Lambda_0 V^*) - f(\Lambda_0) - C\frac{\log^2 n}{n}.
	\end{equation*}
	It remains to deduce from Lemma \ref{lem:f(UL V^*) expansion} that $f(U\Lambda_0 V^*) - f(\Lambda_0) = O(n^{-1})$ uniformly in $U$ and $V$.
	\looseness=-1
	
\end{proof}

Lemma \ref{lem:est for Re f} yields
\begin{equation*}
	\begin{split}
		\CF_m ={}& Cn^{m^2}e^{nf(\Lambda_0)} \Bigg(\int\limits_{\stpointsnbh} \Vanddet^2(\Lambda^2) \prod\limits_{j = 1}^m \lambda_j \times e^{n\tilde{f}(U\Lambda V^*)} d\mu(U) d\mu(V) d\Lambda + O(e^{-C_1\log^2 n})\Bigg),
	\end{split}
\end{equation*}
where $\stpointsnbh$ is defined in \eqref{stpoinnbh def}.
Changing the variables $\Lambda = \Lambda_0 + \frac{1}{\sqrt{n}}\tilde{\Lambda}$ and expanding $f$ according to Lemma \ref{lem:f(UL V^*) expansion} we obtain
\begin{equation}\label{n-indep int}
\begin{split}
\CF_m = C\mathsf{k}_n \int\limits_{\sqrt{n}\stpointsnbh} &\Vanddet^2(\tilde{\Lambda}) \exp\left\{-\tr(2\lambda_0\tilde{\Lambda} + \bar{z}_0\cMatPos_U + z_0\cMatPos_V^*)^2/2 + \tr\cMatPos_U\cMatPos_V^*\right\} \\
&\times d\mu(U) d\mu(V) d\tilde{\Lambda}(1 + o(1)),
\end{split}
\end{equation}
where
\begin{equation}\label{K_n def}
\mathsf{k}_n = n^{m^2/2}e^{-mn\lambda_0^2 + \sqrt{n}\tr \left(\bar{z}_0\cMatPos + z_0\cMatPos^*\right)}.
\end{equation}
Let us change the variables $V = WU$. Taking into account that the Haar measure is invariant w.r.t.\ shifts we get
\begin{equation*}
	\begin{split}
		\CF_m = C\mathsf{k}_n \int\limits_{\R^m} \int\limits_{U(m)} \int\limits_{U(m)} &\Vanddet^2(\tilde{\Lambda}) \exp\left\{-\tr(2\lambda_0\tilde{\Lambda} + U^*(\bar{z}_0\cMatPos + z_0\cMatPos_W^*)U)^2/2 + \tr\cMatPos W^*\cMatPos^* W\right\} \\
		&\times d\mu(U) d\mu(W) d\tilde{\Lambda}(1 + o(1)) \\
		= C\mathsf{k}_n \int\limits_{\R^m} \int\limits_{U(m)} \int\limits_{U(m)} &\Vanddet^2(\tilde{\Lambda}) \exp\left\{-\tr(2\lambda_0U\tilde{\Lambda}U^* + (\bar{z}_0\cMatPos + z_0\cMatPos_W^*))^2/2 + \tr\cMatPos W^*\cMatPos^* W\right\} \\
		&\times d\mu(U) d\mu(W) d\tilde{\Lambda}(1 + o(1)).
	\end{split}
\end{equation*}
The next step is to change the variables $H = U\tilde{\Lambda}U^*$. The Jacobian is $\frac{\prod_{j = 1}^{m - 1} j!}{(2\pi)^{m(m - 1)/2}}\Vanddet^{-2}(\tilde{\Lambda})$ (see e.g.\ \cite{Hu:63}). Thus
\begin{equation*}
	\begin{split}
		\CF_m = C\mathsf{k}_n \int\limits_{\herm_m} \int\limits_{U(m)} & \exp\left\{-\tr(2\lambda_0H + (\bar{z}_0\cMatPos + z_0\cMatPos_W^*))^2/2 + \tr\cMatPos W^*\cMatPos^* W\right\} \\
		&\times d\mu(W) dH(1 + o(1)),
	\end{split}
\end{equation*}
where $\herm_m$ is a space of hermitian $m \times m$ matrices and
\begin{equation*}
	dH = \prod\limits_{j = 1}^m d(H)_{jj} \prod\limits_{j < k} d\Re (H)_{jk} d\Im (H)_{jk}.
\end{equation*}
The Gaussian integration over $H$ implies
\begin{equation}\label{last asympt}
\CF_m = C\mathsf{k}_n \int\limits_{U(m)} \exp\left\{ \tr\cMatPos W^*\cMatPos^* W\right\} d\mu(W) (1 + o(1)).
\end{equation}
If $\cMatPos = 0$, \eqref{last asympt} immediately yields \eqref{main:moments}. Otherwise, for computing the integral over the unitary group, the following Harish-Chan\-dra/It\-syk\-son--Zuber formula is used
\begin{prop}\label{pr:H-C/I--Z formula}
	Let $A$ and $B$ be normal $d \times d$ matrices with distinct eigenvalues $\{a_j\}_{j = 1}^d$ and $\{b_j\}_{j = 1}^d$ respectively. 
	Then
	\begin{equation*} 
		\int\limits_{U(d)} \exp\{z\tr AU^*BU\}d\mu(U) = \bigg(\prod\limits_{j = 1}^{d - 1} j!\bigg) \frac{\det\{\exp(za_jb_k)\}_{j,k = 1}^d}{z^{(d^2 - d)/2}\Vanddet(A)\Vanddet(B)},
	\end{equation*}
	where $z$ is some constant, $\mu$ is a Haar measure, and $\Vanddet(A) = \prod\limits_{j > k}(a_j - a_k)$.
\end{prop}
For the proof see, e.g., \cite[Appendix 5]{Me:91}.
Applying the Harish-Chan\-dra/It\-syk\-son--Zuber formula to \eqref{last asympt} we obtain
\begin{equation*} 
	\CF_m = C\mathsf{k}_n \frac{\det \{e^{\zeta_j\cConjScl{\zeta}_k}\}_{j,k = 1}^m}{\abs{\Vanddet(\cMatPos)}^2} (1 + o(1)),
\end{equation*}
which in combination with \eqref{F_1 behavior} yields the result of Theorem \ref{th1}.
\subsection{General case}
In the general case the proof proceeds by the same scheme as in the Gaussian case. In this subsection we focus on the crucial distinctions from the Gaussian case and refine the corresponding assertions from previous subsection. Set
\begin{equation*}
	\normsized{\hat{\cSetMatA}} = \sum\limits_{\substack{p + s \text{ is even} \\ 0 \le p,s \le m \\ (p,s) \ne (1,1)}} \norm{Q_{p,s}}.
\end{equation*}
The generalization of Lemma \ref{lem:f(UL V^*) expansion} is
\begin{lem}
	Let $\normsized{\tilde{\Lambda}} + \normsized[\big]{\hat{\tilde{\cSetMatA}}} \le \log n$. Then uniformly in $U$ and $V$
	\begin{equation} \label{f expansion gen}
	\begin{split}
	f(U(\Lambda_0 &+ n^{-1/2}\tilde{\Lambda})V^*, n^{1/2}\hat{\tilde{\cSetMatA}}) \\
	={} &- m\lambda_0^2 + n^{-1/2} \tr (\bar{z}_0\cMatPos + z_0\cMatPos^*) - n^{-1} \tr(2\lambda_0\tilde{\Lambda} + \bar{z}_0\cMatPos_U + z_0\cMatPos_V^*)^2/2 \\
	&+ n^{-1} \tr\cMatPos_U\cMatPos_V^* + n^{-1}\lambda_0^2\sqrt{\cumul{2}{2}}\tr(\wedge^2VU^*)\tilde{Q}_2 + n^{-1}\lambda_0^2\sqrt{\cumul{2}{2}}\tr\tilde{Q}_2^*(\wedge^2UV^*) \\
	&- n^{-1}\sum\limits_{\substack{p + s \text{ is even} \\ 0 \le p,s \le m \\ (p,s) \ne (1,1)}} \tr \tilde{Q}_{p,s}^*\tilde{Q}_{p,s} + O\big(n^{-3/2}\log^3 n\big),
	\end{split}
	\end{equation}
	where $\cMatPos_B$ is defined in \eqref{zeta_B def} and $\wedge^2 B$ is the second exterior power of a linear operator $B$ (see \cite{Vi:03} for definition and properties of an exterior power of a linear operator).
\end{lem}
\begin{proof}
	Differently from the Gaussian case $f$ has additional terms of the form $\tr Q_{p,s}^*Q_{p,s}$ and additional term $n^{-1/2} \tilde{h}(Q_2) + n^{-1}\mathtt{p}_c(\hat{\cSetMatA})$ under the logarithm (cf.\ \eqref{f def} and \eqref{f Gauss}), where $\tilde{h}$ and $\mathtt{p}_c$ are defined in the assertion of Lemma \ref{lem:int over psi,tau}. The contribution of the terms $\tr Q_{p,s}^*Q_{p,s}$ to the expansion \eqref{f expansion gen} is evident. Furthermore, $n^{-1}\mathtt{p}_c(n^{-1/2}\hat{\tilde{\cSetMatA}}) = O\big(n^{-3/2}\log^3 n\big)$ because $\mathtt{p}_c$ is a polynomial with zero constant term. Hence, it remains to determine the contribution of the term $n^{-1/2} \tilde{h}(Q_2)$.
	
	In order to simplify notations, let us omit index $k$ in \eqref{tilde h def}. Thus, now $\aVecBt$ and $\aVecCt$ denote vectors 
	\begin{equation*}
		\begin{pmatrix}
			\aSclB_{k1} \\
			\vdots \\
			\aSclB_{km}
		\end{pmatrix}
		\quad \text{and} \quad
		\begin{pmatrix}
			\aSclC_{k1} \\
			\vdots \\
			\aSclC_{km}
		\end{pmatrix}
	\end{equation*}
	respectively. Then \eqref{tilde h def} is written as
	\begin{equation*}
		\tilde{h}(Q_2) = -\int \left(\tr \tilde{\matA}_{2,2} Q_{2} + \tr Q_{2}^* \matA_{2,2}\right) e^{b_{2}} d\aConjMat{\aVecBt} d\aVecB d\aConjMat{\aVecC} d\aVecC,
	\end{equation*}
	where $\tilde{\matA}_{2,2}$ and $\matA_{2,2}$ are defined in \eqref{chi def} and $b_2$ has the form \eqref{new form of b_2}. Therefore
	\begin{multline}\label{tilde h before cv}
		n^{-1/2} \tilde{h}(n^{-1/2}\tilde{Q}_2) = n^{-1} \tilde{h}(\tilde{Q}_2) = -\frac{\sqrt{\cumul{2}{2}}}{n}\int d\aConjMat{\aVecBt} d\aVecB d\aConjMat{\aVecC} d\aVecC\, e^{- \aConjMat{\aVecD}A\aVecD} \\
		\times\sum\limits_{\alpha, \beta \in \indexset_{m,2}} \left(\aSclC_{k\beta_{1}}^{\phantom{+}} \aSclC_{k\beta_{2}}^{\phantom{+}} \aConjScl{\aSclB_{k\alpha_{1}}}\aConjScl{\aSclB_{k\alpha_{2}}} \tilde{q}_{\alpha\beta}^{(2)} + \bar{\tilde{q}}_{\alpha\beta}^{(2)} \aSclB_{k\alpha_{1}} \aSclB_{k\alpha_{2}} \aConjScl{\aSclC_{k\beta_{1}}}\aConjScl{\aSclC_{k\beta_{2}}}\right),
	\end{multline}
	where $\aVecD$ is defined in \eqref{Psi def}, $\indexset_{m,2}$ is defined in \eqref{indexset def}.	Let us change the variables $\tilde{\aVecB} = U^*\aVecBt$, $\aConjMat{\tilde{\aVecB}} = \aConjMat{\aVecBt}U$, $\tilde{\aVecC} = V^*\aVecC$, $\aConjMat{\tilde{\aVecC}} = \aConjMat{\aVecC}V$. We have
	\begin{equation}\label{cv1}
	\begin{split}
	\aSclC_{k\beta_{1}}^{\phantom{+}} \aSclC_{k\beta_{2}}^{\phantom{+}} \aConjScl{\aSclB_{k\alpha_{1}}}\aConjScl{\aSclB_{k\alpha_{2}}} &= (V\tilde{\aVecC})_{\beta_{1}}(V\tilde{\aVecC})_{\beta_{2}} (\aConjMat{\tilde{\aVecB}}U^*)_{\alpha_{1}}(\aConjMat{\tilde{\aVecB}}U^*)_{\alpha_{2}} \\
	&= \sum\limits_{\gamma_1, \gamma_2 = 1}^m \sum\limits_{\delta_1, \delta_2 = 1}^m v_{\beta_1\gamma_1}\tilde{\aSclC}_{k\gamma_1}v_{\beta_{2}\gamma_2} \tilde{\aSclC}_{k\gamma_2}\aConjScl{\tilde{\aSclB}_{k\delta_1}} \conj{u}_{\alpha_{1}\delta_1}\aConjScl{\tilde{\aSclB}_{k\delta_2}} \conj{u}_{\alpha_{2}\delta_2} \\
	&= \sum\limits_{\gamma, \delta \in \indexset_{m,2}} (v_{\beta_1\gamma_1}v_{\beta_{2}\gamma_2} - v_{\beta_1\gamma_2}v_{\beta_{2}\gamma_1})\tilde{\aSclC}_{k\gamma_1} \tilde{\aSclC}_{k\gamma_2}\aConjScl{\tilde{\aSclB}_{k\delta_1}} \aConjScl{\tilde{\aSclB}_{k\delta_2}} \\
	&\phantom{= \sum\limits_{\gamma, \delta \in \indexset_{m,2}}}\times (\conj{u}_{\alpha_{1}\delta_1} \conj{u}_{\alpha_{2}\delta_2} - \conj{u}_{\alpha_{1}\delta_2}\conj{u}_{\alpha_{2}\delta_1}) \\
	&= \sum\limits_{\gamma, \delta \in \indexset_{m,2}} (\wedge^2V)_{\beta\gamma}\tilde{\aSclC}_{k\gamma_1} \tilde{\aSclC}_{k\gamma_2}\aConjScl{\tilde{\aSclB}_{k\delta_1}} \aConjScl{\tilde{\aSclB}_{k\delta_2}}(\wedge^2U^*)_{\delta\alpha},
	\end{split}
	\end{equation}
	where $u_{jk} = (U)_{jk}$, $v_{jk} = (V)_{jk}$. Similarly
	\begin{equation}\label{cv2}
	\aSclB_{k\alpha_{1}} \aSclB_{k\alpha_{2}} \aConjScl{\aSclC_{k\beta_{1}}}\aConjScl{\aSclC_{k\beta_{2}}} = \sum\limits_{\gamma, \delta \in \indexset_{m,2}} (\wedge^2U)_{\alpha\gamma} \tilde{\aSclB}_{k\gamma_1} \tilde{\aSclB}_{k\gamma_2} \aConjScl{\tilde{\aSclC}_{k\delta_1}} \aConjScl{\tilde{\aSclC}_{k\delta_2}} (\wedge^2V^*)_{\delta\beta}.
	\end{equation}
	Besides,
	\begin{equation}\label{cv3}
	\aConjMat{\aVecD}A\aVecD = \aConjMat{\tilde{\aVecD}}\tilde{A}\tilde{\aVecD} = \aConjMat{\tilde{\aVecD}}A_0\tilde{\aVecD} + O(n^{-1/2}\log n),
	\end{equation}
	where $A_0$ is defined in \eqref{A_0,A_1 def} and
	\begin{align*}
		\tilde{\aVecD} = 
		\begin{pmatrix}
			\tilde{\aVecB} \\
			\tilde{\aVecC}
		\end{pmatrix}, 
		\qquad \tilde{A} = 
		\begin{pmatrix}
			U^*	& 0 \\
			0	& V^*
		\end{pmatrix}
		A
		\begin{pmatrix}
			U	& 0 \\
			0	& V
		\end{pmatrix} 
		= 
		\begin{pmatrix}
			-U^*ZU	 & \Lambda \\
			-\Lambda & -V^*Z^*V
		\end{pmatrix}.
	\end{align*}
	The ``differentials'' change as follows 
	\begin{equation}\label{cv_diff}
	\begin{split}
	d\aVecB &= \det\nolimits^{-1} U d\tilde{\aVecB}, \\
	d\aConjMat{\aVecB} &= \det\nolimits^{-1} U^* d\aConjMat{\tilde{\aVecB}}
	\end{split}
	\end{equation}
	and for $d\aVecC$ likewise. Eventually, substitution of \eqref{cv1}--\eqref{cv_diff} into \eqref{tilde h before cv} yields
	\begin{equation} \label{tilde h after cv}
	\begin{split}
	n^{-1} \tilde{h}(\tilde{Q}_2) &= -\frac{\sqrt{\cumul{2}{2}}}{n} \sum\limits_{\gamma, \delta \in \indexset_{m,2}} \int \bigg(\tilde{\aSclC}_{k\gamma_1} \tilde{\aSclC}_{k\gamma_2}\aConjScl{\tilde{\aSclB}_{k\delta_1}} \aConjScl{\tilde{\aSclB}_{k\delta_2}}\left((\wedge^2U^*)\tilde{Q}_2 (\wedge^2V)\right)_{\delta\gamma} \\
	&\phantom{=-\frac{\sqrt{\cumul{2}{2}}}{n}\int \sum\limits_{\gamma, \delta \in \indexset_{m,2}} \bigg(}+ \left((\wedge^2V^*)\tilde{Q}_2^* (\wedge^2U)\right)_{\delta\gamma} \tilde{\aSclB}_{k\gamma_1}\tilde{\aSclB}_{k\gamma_2}\aConjScl{\tilde{\aSclC}_{k\delta_1}} \aConjScl{\tilde{\aSclC}_{k\delta_2}}\bigg) \\
	&\phantom{=-\frac{\sqrt{\cumul{2}{2}}}{n}\sum\limits_{\gamma, \delta \in \indexset_{m,2}}\int}\times e^{-\aConjMat{\tilde{\aVecD}}A_0\tilde{\aVecD}} d\aConjMat{\tilde{\aVecB}} d\tilde{\aVecB} d\aConjScl{\tilde{\aVecC}} d\tilde{\aVecC} + O\big(n^{-3/2}\log^3 n\big)
	\end{split}
	\end{equation}
	uniformly in $U$ and $V$. Due to the structure of $A_0$ the integration can be performed over $\tilde{\aSclB}_{kj}$, $\tilde{\aSclC}_{kj}$ separately for every~$j$.
	Notice that
	\begin{equation*}
		\int \aSclE \exp\left\{-\aConjMat{\aVecF} A'\aVecF\right\} d\aConjScl{\tilde{\aSclB}_{kj}} d\tilde{\aSclB}_{kj} d\aConjScl{\tilde{\aSclC}_{kj}} d\tilde{\aSclC}_{kj} = 0,
	\end{equation*}
	where $\aSclE$ is either $\aConjScl{\tilde{\aSclB}_{kj}}$, $\tilde{\aSclB}_{kj}$, $\aConjScl{\tilde{\aSclC}_{kj}}$ or $\tilde{\aSclC}_{kj}$ and
	\begin{equation*}
		\aVecF =
		\begin{pmatrix}
			\tilde{\aSclB}_{kj} \\
			\tilde{\aSclC}_{kj}
		\end{pmatrix},
		\quad A' =
		\begin{pmatrix}
			-z_0	 & \lambda_0 \\
			-\lambda_0 & -\bar{z}_0
		\end{pmatrix}.
	\end{equation*}
	Hence the terms with $\gamma \ne \delta$ in \eqref{tilde h after cv} are zeros. Furthermore, expanding the exponent into series, one can observe that
	\begin{equation*}
		\int \tilde{\aSclC}_{kj}\aConjScl{\tilde{\aSclB}_{kj}} e^{-\aConjMat{\aVecF} A'\aVecF} d\aConjScl{\tilde{\aSclB}_{kj}} d\tilde{\aSclB}_{kj} d\aConjScl{\tilde{\aSclC}_{kj}} d\tilde{\aSclC}_{kj} = -\int \tilde{\aSclB}_{kj}\aConjScl{\tilde{\aSclC}_{kj}} e^{-\aConjMat{\aVecF} A'\aVecF} d\aConjScl{\tilde{\aSclB}_{kj}} d\tilde{\aSclB}_{kj} d\aConjScl{\tilde{\aSclC}_{kj}} d\tilde{\aSclC}_{kj} = \lambda_0.
	\end{equation*}
	It implies
	\begin{equation*}
		\begin{split}
			n^{-1} \tilde{h}(\tilde{Q}_2) &= n^{-1}\lambda_0^2\sqrt{\cumul{2}{2}} (\tr(\wedge^2U^*)\tilde{Q}_2(\wedge^2V) + \tr(\wedge^2V^*)\tilde{Q}_2^*(\wedge^2U)) + o(n^{-1}) \\
			&= n^{-1}\lambda_0^2\sqrt{\cumul{2}{2}} (\tr(\wedge^2VU^*)\tilde{Q}_2 + \tr\tilde{Q}_2^*(\wedge^2UV^*)) + O\big(n^{-3/2}\log^3 n\big).
		\end{split}
	\end{equation*}
	The above relation completes the proof of \eqref{f expansion gen}.
\end{proof}
An analog of Lemma \ref{lem:est for Re f} is
\begin{lem} 
	Let $\tilde{f}(\cSetMatA) = f(\cSetMatA) - f(\Lambda_0, 0)$. Then for sufficiently large $n$
	\begin{equation*} 
		\max_{\frac{\log n}{\sqrt{n}} \le \norm{\Lambda - \Lambda_0} + \normsized{\hat{\cSetMatA}} \le r} \Re \tilde{f}(U\Lambda V^*, \hat{\cSetMatA}) \le -C\frac{\log^2 n}{n}
	\end{equation*}
	uniformly in $U$ and $V$.
\end{lem}
The proof needs only cosmetic changes because of additional variables $\hat{\cSetMatA}$. Following the proof in the Gaussian case one can see that \eqref{n-indep int} transforms into
\begin{equation*}
	\begin{split}
		\CF_m = C\mathsf{k}_n \int\limits_{\sqrt{n}\stpointsnbh} &\Vanddet^2(\tilde{\Lambda}) \exp\Big\{-\tr(2\lambda_0\tilde{\Lambda} + \bar{z}_0\cMatPos_U + z_0\cMatPos_V^*)^2/2 + \tr\cMatPos_U\cMatPos_V^* \\
		&+ \lambda_0^2\sqrt{\cumul{2}{2}}\tr(\wedge^2VU^*)\tilde{Q}_2 + \lambda_0^2\sqrt{\cumul{2}{2}}\tr\tilde{Q}_2^*(\wedge^2UV^*) \\
		&- \sum\limits_{\substack{p + s \text{ is even} \\ 0 \le p,s \le m \\ (p,s) \ne (1,1)}} \tr \tilde{Q}_{p,s}^*\tilde{Q}_{p,s}\Big\} d\mu(U) d\mu(V) d\tilde{\Lambda}d\hat{\cSetMatA}(1 + o(1)),
	\end{split}
\end{equation*}
where $\mathsf{k}_n$ is defined in \eqref{K_n def}. The Gaussian integration over $\hat{\cSetMatA}$ yields
\begin{equation*}
	\begin{split}
		\CF_m = C\mathsf{k}_n\exp\left\{\frac{m^2 - m}{2}\lambda_0^4\cumul{2}{2}\right\} \int &\Vanddet^2(\tilde{\Lambda}) \exp\Big\{-\tr(2\lambda_0\tilde{\Lambda} + \bar{z}_0\cMatPos_U + z_0\cMatPos_V^*)^2/2 + \tr\cMatPos_U\cMatPos_V^* \Big\} \\
		&\times d\mu(U) d\mu(V) d\tilde{\Lambda}(1 + o(1)).
	\end{split}
\end{equation*}
The last formula shows that there are no differences in further proof up to a high moments independent factor $\exp\left\{\frac{m^2 - m}{2}\lambda_0^4\cumul{2}{2}\right\}$.

\bibliography{GinDetCorr}


\end{document}